\documentclass[letterpaper, 10 pt, conference]{ieeeconf}  

\IEEEoverridecommandlockouts                              

\usepackage{graphicx}
\usepackage{mathptmx} 
\usepackage{times} 
\usepackage{amsmath} 
\usepackage{amssymb}  
\usepackage{cite}
\usepackage{algorithm}
\usepackage{algpseudocode}
\usepackage{url}

\usepackage{booktabs}

\newtheorem{remark}{Remark}
\newtheorem{defn}{Definition}
\newtheorem{prop}{Proposition}
\newtheorem{assumption}{Assumption}
\newtheorem{theorem}{Theorem}
\newtheorem{lemma}{Lemma}

\newcommand{\R}{\mathbb{R}}
\newcommand{\norm}[1]{\left\lVert #1 \right\rVert}

\newcommand{\rank}{\operatorname{rank}}
\newcommand{\sym}{\operatorname{sym}}

\newcommand{\GL}{\operatorname{GL}}
\newcommand{\Ogroup}{\mathcal{O}}
\newcommand{\score}{S}
\newcommand{\design}{\mathcal{K}}

\title{\LARGE \bf
Anomaly Score Preserving Data Poisoning Attacks on\\
Data-Driven Controller Design
}

\author{Motoharu Tsuboi, Hampei Sasahara, Iori Takaki, and Hideaki Ishii
\thanks{*This work was supported in part by JST-ASPIRE Program Grant Number JPMJAP2402 and JSPS KAKENHI Grant Number JP24K17296 and JST PRESTO Grant Number JPMJPR26K3.}
\thanks{The authors are with the Graduate School of Information Science and Technology,
The University of Tokyo, Tokyo 113-8656, Japan. 
{\tt\small Email:\{tsuboharu1202, hsasahara, ioritakaki\}@g.ecc.u-tokyo.ac.jp, hideaki\_ishii@ipc.i.u-tokyo.ac.jp}.}%
}

\begin{document}

\maketitle
\thispagestyle{empty}
\pagestyle{empty}

\begin{abstract}

This paper investigates stealthy data poisoning attacks against data-driven controller synthesis and develops an attack generation method.
During data collection, an attacker can sequentially manipulate state measurements, while the defender attempts to detect such manipulation using a residual-based anomaly score based on the Mahalanobis distance.
Under suitable assumptions on the defender's prior knowledge, we characterize the class of time-varying affine transformations that preserve the anomaly score for arbitrary data and parameterize this class by the orthogonal group.
Based on this characterization, we formulate attack design under the stealthiness condition as an optimization problem over the orthogonal group and propose a Riemannian gradient method.
Numerical experiments demonstrate that the proposed attack can degrade closed-loop performance while remaining stealthy.
\end{abstract}

\section{Introduction}
In control system design, constructing controllers from collected data has become a standard approach for systems with unknown dynamics.
The indirect approach identifies a system model from data and then designs a controller based on it~\cite{ljung1999system,vanoverschee1996subspace}.
More recently, 
direct data-driven control, which designs controllers directly from data without explicitly identifying a model, has received increasing attention~\cite{willems2005persistency,depersis2020formulas,vanwaarde2020data,Dorfler2023}.
In either approach, the resulting controller depends critically on the collected data.
This dependence can be exploited by cyberattacks on the collected data, providing an emerging attack vector against control systems.
In particular, malicious manipulation of data used for controller design, commonly referred to as a data poisoning attack, can have severe consequences.
Threat assessment and the development of countermeasures against data poisoning attacks have therefore become important research topics in data-driven control~\cite{russo2021poisoning,Kaminaga2024CCTA,sasahara2026adversarial,Anand2025arxiv,digge2026systematic,takaki2026datapoisoningattacksinformativity}, and the practical significance of such threats has also been demonstrated in applications such as autonomous driving systems~\cite{burbano2026badcontrol}.

To properly assess the threat posed by sophisticated data poisoning attacks, it is essential to consider not only degradation of control performance but also evasion of detection, i.e., stealthiness.
Previous studies have evaluated attack stealthiness by measuring the magnitude of data manipulation relative to the original data using matrix norms~\cite{sasahara2026adversarial,russo2021poisoning}.
In practice, however, the designer does not have access to the unmanipulated data, and a bound on the manipulation magnitude alone does not generally guarantee evasion of the detection mechanism used by the designer.
Moreover, analyses based on norm constraints may overlook sophisticated attacks that substantially manipulate the data without affecting the detection statistic, thereby underestimating the range of potential threats.
Attack stealthiness should therefore be defined rigorously in relation to the anomaly-detection statistic actually employed by the defender.
This issue has been studied for least-squares identification of linear dynamical systems, where detectability is analyzed using residual norms and $F$-test statistics~\cite{russo2023detectability}.
Another issue concerns the availability of data to the attacker.
Attack methods that require the entire dataset in advance are not directly applicable to attacks that sequentially manipulate previously unseen sensor data.

The objective of this study is to develop a framework for evaluating data poisoning attacks against data-driven controller design that is based on a specific detection statistic and allows sequential manipulation of previously unseen data.
In this study, we employ the Mahalanobis distance, which is widely used for anomaly detection in multivariate data~\cite{JohnsonWichern2007,Qin2012Survey,Venkatasubramanian2003PartIII,pasqualetti2013attack,giraldo2018physics}, for defining the anomaly score.
Specifically, the anomaly score is defined as the infimum, over candidate models consistent with the designer's prior knowledge, of the sum of squared Mahalanobis distances between the observed states and their one-step predictions.
Under the detection scheme, we consider an attack scenario in which the attacker knows the true system model and compromises the sensor communication channel during data collection. The attacker sequentially manipulates the state data through a time-varying affine transformation to degrade the resulting control performance while remaining stealthy.

To obtain a stealthiness guarantee for previously unseen data, we define \emph{$S$-invariant stealthiness} as the property of a data transformation that preserves the anomaly score for arbitrary data.
First, under suitable assumptions on the prior knowledge, we show that a time-varying affine transformation is $S$-invariant stealthy if and only if it reduces to a time-invariant weighted orthogonal transformation that also preserves the candidate model set.
Based on the finding, we formulate the generation of manipulated data as an optimization problem over the orthogonal group and propose a gradient-based attack-generation method on a Riemannian manifold.
Furthermore, we introduce a regularization term into the objective function to suppress excessive data manipulation that may be noticed by other detection mechanisms.
Numerical experiments on a standard benchmark model demonstrate that the proposed method can significantly destabilize the closed-loop system while preserving the anomaly score. Specifically, it increases the spectral radius of the closed-loop matrix to 2.674 while keeping the anomaly score nearly zero ($2.34 \times 10^{-23}$). In contrast, an existing method achieves a spectral radius of 1.234 with a comparable amount of data manipulation, but increases the anomaly score to $2.21 \times10^4$.

The main contributions of this paper are summarized as follows.
First, we formulate data-independent stealthiness, referred to as $S$-invariant stealthiness, based on a concrete anomaly-detection statistic, and characterize its necessary and sufficient conditions.
Second, we develop a practical attack generation method based on Riemannian gradient optimization.
Third, we introduce and assess an attack scenario in which previously unseen data flowing through a sensor communication channel are sequentially manipulated during the data-collection stage. Our characterization enables the sequential application of a precomputed attack transformation to previously unseen data, guaranteeing anomaly-score preservation independently of the data used for attack optimization.

This paper is organized as follows.
Section~\ref{sec:system_and_data} introduces the system, the controller-design procedure, and the attacker model.
Section~\ref{sec:stealth} characterizes $S$-invariant stealthy transformations.
Section~\ref{sec:optimization} presents the attack design algorithm.
Section~\ref{sec:numerical} provides numerical experiments, and Section~\ref{sec:conclusion} concludes the paper.

\indent \textit{Notation:}~For a square matrix $M$, let $\rho(M)$ denote its spectral radius.
The Frobenius norm of a matrix $M$ is denoted by $\|M\|_{\mathrm{F}}$, and its elementwise maximum norm is defined as $\|M\|_{\max}:=\max_{i,j}|M_{ij}|$.
For $M\in\mathbb{R}^{n\times T},$ and a symmetric positive definite matrix $\Sigma\in\mathbb{R}^{n\times n}$,
define $\|M\|_{\Sigma^{-1},\mathrm{F}}:= \sqrt{\operatorname{tr}\!\left(M^\top\Sigma^{-1}M\right)}$.
We use $\GL(n)$ and $\mathrm{O}(n)$ to denote the general linear group of degree $n$ and the orthogonal group, respectively, and define 
$\operatorname{sym}(M):=(M+M^\top)/2$.
For $A\in\mathbb R^{n\times n}$ and $\varepsilon>0$,
let $B_\varepsilon(A)$ denote the open ball of radius $\varepsilon$ centered at $A$.
For a set $\mathcal S\subseteq\mathbb R^{n\times n}$,
its closure is denoted by $\overline{\mathcal S}$.

\section{Controller Design and Anomaly Detection}
\label{sec:system_and_data}
\subsection{Controller Design}
Consider the discrete-time linear time-invariant system:
\begin{equation}
 x_{t+1}=Ax_t+Bu_t+w_t,
 \qquad t=0,1,\ldots,T-1
 \label{eq:system}
\end{equation}
where $x_t\in\R^n$ is the state, $u_t\in\R^m$ is the input, and $w_t\in\R^n$ is the process noise.
The matrices $A\in\R^{n\times n}$ and $B\in\R^{n\times m}$ are the true system matrices.
The process noise is assumed to follow a zero-mean Gaussian distribution with covariance $\Sigma\succ 0$, i.e., $w_t\sim\mathcal{N}(0,\Sigma)$.
Define the state and input data matrices collected from the system as
\begin{align}
 X&:=\begin{bmatrix}x_0&x_1&\cdots&x_{T-1}\end{bmatrix}\in\R^{n\times T},\nonumber\\
 Z&:=\begin{bmatrix}x_1&x_2&\cdots&x_T\end{bmatrix}\in\R^{n\times T},\nonumber\\
 U&:=\begin{bmatrix}u_0&u_1&\cdots&u_{T-1}\end{bmatrix}\in\R^{m\times T}
 \label{eq:data_matrices}
\end{align}
and denote them collectively by $D:=(X,Z,U)$.
Also, define the process noise as
\begin{equation}
 W:=\begin{bmatrix}w_0&w_1&\cdots&w_{T-1}\end{bmatrix}\in\R^{n\times T}.
 \label{eq:noise_matrix}
\end{equation}
Then, from \eqref{eq:system}, these data matrices satisfy
\begin{equation}
 Z=AX+BU+W.
 \label{eq:data_dynamics_with_noise}
\end{equation}

We consider a setting in which the controller designer does not have complete knowledge of the true system matrices $(A,B)$ and designs a state-feedback controller using the collected data $D$.

We represent the controller design rule employed by the designer as the mapping
\[
 \design:\mathbb{R}^{n\times T}\times \mathbb{R}^{n\times T} \times \mathbb{R}^{m\times T}\to\mathbb{R}^{m\times n}
\]
and define the control input as $u_t=Kx_t$, where the feedback gain is given by $K:=\design(D)$.

\begin{remark}
Data-based controller design includes indirect and direct approaches~\cite{Dorfler2023}.
The former identifies a system model before designing a controller, while the latter constructs a controller directly from data.
We consider both approaches and denote the general mapping from data $D$ to a feedback gain $K$ by $\design$.
\end{remark}

\subsection{Anomaly Detection}
Next, we formulate a prior-knowledge-based anomaly detection method for evaluating whether the collected data are consistent with the designer's prior knowledge.

\begin{assumption}[Prior Knowledge]
\label{ass:prior_knowledge}
The designer has prior knowledge of the true system matrices
in the form
\begin{equation}
 A\in\mathcal{A}\subseteq\mathbb{R}^{n\times n},
 \qquad
 B\in\mathcal{B}\subseteq\mathbb{R}^{n\times m},
 \label{eq:prior_sets}
\end{equation}
where $\mathcal{A}$ and $\mathcal{B}$ are known closed sets.
\end{assumption}

A pair $(A_c,B_c)\in\mathcal{A}\times\mathcal{B}$
is referred to as a candidate model.
Combining prior system knowledge with data-driven approaches has been widely studied as a natural framework in recent data-driven control and modeling.
Examples include system-theoretic properties, such as stabilizability and controllability, as well as structural constraints~\cite{berberich2023combining,shakouri2025prior,engeln2026prior,iwata2025positive}.

For data $D=(X,Z,U)$ and a candidate model $(A_c,B_c)$, we evaluate the residual at each time instant using the Mahalanobis distance induced by the process-noise covariance $\Sigma$. The overall residual magnitude is then given by $\left\|Z-A_cX-B_cU\right\|_{\Sigma^{-1},\mathrm{F}}^2$.
We define the infimum of this quantity over all candidate models as the anomaly score of the data $D$.

\begin{defn}[Anomaly Score]
For data $D=(X,Z,U)$, define the anomaly score as
\begin{equation}
 S(D)
 :=
 \inf_{A_c\in\mathcal{A},\,B_c\in\mathcal{B}}
 \left\|Z-A_cX-B_cU\right\|_{\Sigma^{-1},\mathrm{F}}^2.
 \label{eq:model_consistency_score}
\end{equation}
\end{defn}

Given a prescribed threshold $\eta>0$, the designer declares the data $D$ anomalous if $S(D)>\eta$.
The quantity in \eqref{eq:model_consistency_score} can be interpreted as the minimum negative log-likelihood over the candidate models under Gaussian noise~\cite{JohnsonWichern2007,yin2023maximum}.
Residual-based evaluations using the Mahalanobis distance are standard in a wide range of fields, including multivariate analysis, machine learning and control engineering
~\cite{JohnsonWichern2007,Qin2012Survey,mo2010false,Venkatasubramanian2003PartIII,lee2018simple,pasqualetti2013attack,giraldo2018physics}.
Recent work on data-driven control evaluates attack detectability using weighted squared residuals under maximum-likelihood models~\cite{russo2023detectability}.

\subsection{Formulation of Data Poisoning Attacks}
\label{sec:attack_scenario}

\begin{figure}[t]
\centering
\includegraphics[width=0.95\linewidth]{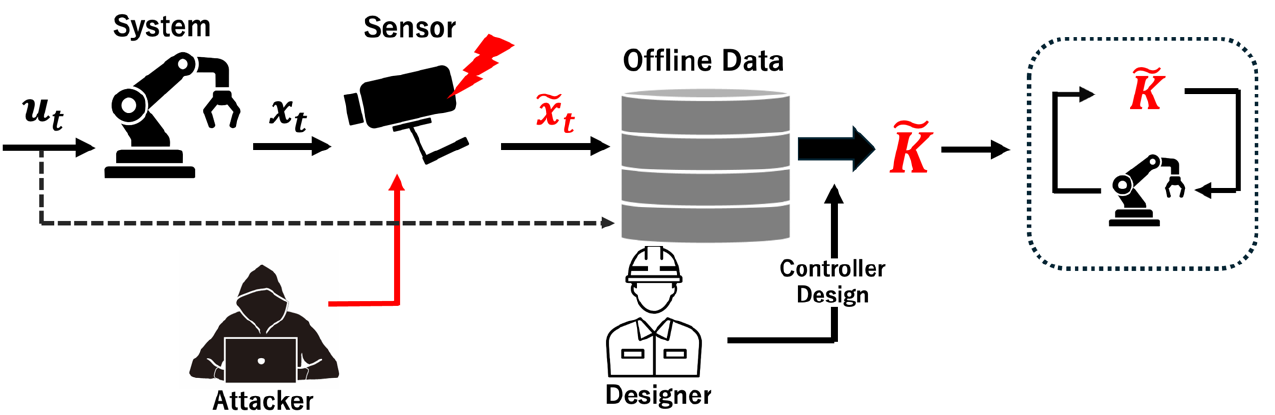}
\caption{Sensor attack during the collection of data for controller design.}
\label{fig:threat_scenario}
\end{figure}

The attack scenario considered in this paper is illustrated in Fig.~\ref{fig:threat_scenario}.
The attacker is assumed to know the true system matrices $(A,B)$ and to be capable of manipulating state data obtained from the sensors.
In contrast, the input signals used for data collection are predetermined by the designer, and therefore the attacker cannot manipulate the input data.
After data collection, the designer constructs a controller from the manipulated data and implements it on the true system.
The attack terminates before the implementation stage.

In this study, the attacker manipulates the state data using the following time-varying affine transformation $\mathcal{T}$:
\begin{equation}
 \mathcal{T}:\widetilde x_t=T_{x,t}x_t+b_{x,t},
 \qquad t=0,1,\ldots,T
 \label{eq:state_affine_attack}
\end{equation}
where $T_{x,t}\in\mathbb{R}^{n\times n}$ and $b_{x,t}\in\mathbb{R}^{n}$ are variables chosen by the attacker.
The data transformed according to \eqref{eq:state_affine_attack} are denoted by $\widetilde D=\mathcal{T}(D)$.

Next, based on the anomaly score introduced in the previous subsection, we define \textit{S}-invariant stealthiness.
\begin{defn}[\textit{S}-Invariant Stealthiness]
A data transformation $\mathcal{T}$ is said to be \textit{S}-invariant stealthy if
\begin{equation}
    \score\bigl(\mathcal{T}(D)\bigr)=\score(D)
    \label{eq:perfect_stealth_definition}
\end{equation}
holds for any finite-length data $D$.
\end{defn}

In the attack scenario considered in this paper, the attacker has access to the sensor data at each time instant but not to the entire dataset in advance.
Thus, the attacker requires a single transformation rule whose structure itself preserves the anomaly score independently of the specific data values.
The above definition reflects this operational constraint.

The attacker’s objective is to degrade the closed-loop performance of the true system.
We denote the attack objective by $J_{\mathrm{atk}}(\design(\mathcal{T}(D));A,B)$.

We assume a white-box setting regarding the attacker's knowledge.
Specifically, the attacker is assumed to know the controller design method, the definition of the anomaly score, the candidate model sets, and the covariance matrix $\Sigma$.
This setting is intended to capture an advanced persistent threat (APT) scenario, in which a sophisticated attacker acquires detailed knowledge through long-term infiltration.

Accordingly, the data poisoning attack generation problem is formulated as the following constrained optimization problem:
\begin{equation}
\begin{aligned}
 \underset{\mathcal{T}}{\operatorname{maximize}}\quad
 & J_{\mathrm{atk}}\!\left(
     \design(\mathcal{T}(D));A,B
   \right)\\
 \operatorname{subject\ to}\quad
 & \mathcal{T}\ \text{is \textit{S}-invariant stealthy.}
\end{aligned}
\label{eq:conceptual_attack_problem}
\end{equation}

\section{\textit{S}-Invariant Stealthy Attacks}
\label{sec:stealth}


We first characterize the \textit{S}-invariant stealthy transformations under the following assumption.

\begin{assumption}
\label{ass:stealth_characterization}
The following conditions are assumed.
\begin{enumerate}
 \item The set $\mathcal{A}$ of system matrices satisfies
 \begin{equation}
  \forall x\in\R^n\setminus\{0\},\quad
  \exists A_1,A_2\in\mathcal A \quad \text{s.t.} \quad
  \ (A_1-A_2)x\neq0.
  \label{eq:no_common_action_direction}
 \end{equation}

 \item The set $\mathcal A\cap\mathcal C_{m+1}$ is dense in $\mathcal A$, where $\mathcal C_{p}$ for $p \geq 1$ is defined as
 \begin{equation}\label{eq:Cm1}
  \begin{aligned}
    \mathcal{C}_{p} := \Big\{\, A\in\R^{n\times n} \;\Big|\;
      & \exists G\in\R^{n\times p} \text{ s.t.}\\
      & (A,G) \text{ is reachable} \,\Big\}.
  \end{aligned}
 \end{equation}
 

 \item No prior knowledge of $B$ is available, i.e.,
 $\mathcal B=\R^{n\times m}$.
\end{enumerate}
\end{assumption}


The following theorem is the main result of this paper, providing a necessary and sufficient characterization of $S$-invariant stealthy affine transformations.
\begin{theorem}[\textit{S}-Invariant Stealthy Affine Transformations]
\label{thm:perfect_stealth}
Under Assumptions~\ref{ass:prior_knowledge} and~\ref{ass:stealth_characterization},
the time-varying affine transformation \eqref{eq:state_affine_attack} is
\textit{S}-invariant stealthy if and only if there exists
$T_x\in\GL(n)$ satisfying the following conditions.
\begin{enumerate}
 \item Time-invariant linear transformation:
 \begin{equation}
  T_{x,t}=T_x,
  \qquad
  b_{x,t}=0,
  \qquad
  t=0,1,\ldots,T.
  \label{eq:time_invariant_sensor_attack}
 \end{equation}

 \item Covariance invariance:
 \begin{equation}
  T_x^\top\Sigma^{-1}T_x=\Sigma^{-1}.
  \label{eq:weighted_isometry_condition}
 \end{equation}

 \item Invariance of the prior-knowledge set:
 \begin{equation}
  \forall A\in\R^{n\times n},\quad
  A\in\mathcal A
  \iff
  T_xAT_x^{-1}\in\mathcal A.
  \label{eq:prior_invariance_plus}
 \end{equation}
\end{enumerate}
\end{theorem}

To prove Theorem~\ref{thm:perfect_stealth}, we first establish several lemmas.
Let $e_1=(1,0,\ldots,0)^\top\in\R^m$, and define the anomaly score with the input fixed as $u_t=e_1$ by
\begin{equation}
 \begin{array}{l}
 S_{e_1,T}(x_0,\ldots,x_T)
 :={}
 \\
 \inf_{A_c\in\mathcal A,\,B_c\in\mathcal B}\sum_{t=0}^{T-1}\|x_{t+1}-A_cx_t-B_ce_1\|_{\Sigma^{-1},\mathrm F}^2.
 \end{array}
 \label{eq:fixed_input_score}
\end{equation}

\begin{lemma}[Invariance under Common Translation]
\label{lem:common_translation}
For any $T\geq2$, any state sequence $(x_0,\ldots,x_T)$,
and any $c\in\R^n$,
\begin{equation}
 S_{e_1,T}(x_0+c,\ldots,x_T+c) =S_{e_1,T}(x_0,\ldots,x_T)
 \label{eq:common_translation}
\end{equation}
holds.
\end{lemma}

\begin{proof}
For any $(A_c,B_c)\in\mathcal A\times\mathcal B$, let
$B_c':=B_c+(I-A_c)c\,e_1^\top$.
Then,
$B_c'e_1=B_ce_1+(I-A_c)c$,
 and hence, for each $t$,
\[
 (x_{t+1}+c)-A_c(x_t+c)-B_c'e_1
 =
 x_{t+1}-A_cx_t-B_ce_1.
\]
Since $\mathcal B=\R^{n\times m}$,
the mapping $B_c\mapsto B_c'$ is a bijection on $\mathcal B$ for any fixed $A_c$.
Taking the infimum over all candidate models therefore yields
\eqref{eq:common_translation}.
\end{proof}

\begin{lemma}[Zero-State Propagation]
\label{lem:zero_propagation}
Let $T\geq2$ and $0\leq i<j\leq T-1$.
If a state sequence $(y_0,\ldots,y_T)$ satisfies
\[
 S_{e_1,T}(y_0,\ldots,y_T)=0,
 \quad
 y_i=y_{i+1}=0,
 \quad
 y_j=0,
\]
then $y_{j+1}=0$.
\end{lemma}

\begin{proof}
For any $B_c\in\mathcal B$, let
$\beta:=B_ce_1$.
Since $y_i=y_{i+1}=y_j=0$, the residuals at times $i$ and $j$ are
$-\beta$ and $y_{j+1}-\beta$, respectively.
Therefore,
\begin{align*}
 S_{e_1,T}(y_0,\ldots,y_T)
 &\geq
 \inf_{\beta\in\R^n}
 \left(
  \|\beta\|_{\Sigma^{-1},\mathrm F}^2
 +
  \|y_{j+1}-\beta\|_{\Sigma^{-1},\mathrm F}^2
 \right)\\
 &=
 \frac12\|y_{j+1}\|_{\Sigma^{-1},\mathrm F}^2.
\end{align*}
Since the left-hand side is zero and $\Sigma\succ0$,
we obtain $y_{j+1}=0$.
\end{proof}

We now prove Theorem~\ref{thm:perfect_stealth} using the above lemmas.
In the proof, the subscript $T$ indicates the data length

\begin{proof}[Theorem~\ref{thm:perfect_stealth}]
(Sufficiency)
Suppose that the three conditions in the theorem hold.
For any $(A_c,B_c)\in\mathcal A\times\mathcal B$, define
$\widehat A_c:=T_xA_cT_x^{-1}$ and $\widehat B_c:=T_xB_c$.
By \eqref{eq:prior_invariance_plus} and Assumption~\ref{ass:stealth_characterization}, this correspondence is a bijection on the candidate model set
$\mathcal A\times\mathcal B$.
Moreover,
\begin{align*}
 \widetilde Z_T
 -\widehat A_c\widetilde X_T
 -\widehat B_cU_T
 &=
 T_x(Z_T-A_c X_T-B_c U_T).
\end{align*}
Hence, by \eqref{eq:weighted_isometry_condition},
\[
 \|\widetilde Z_T
 -\widehat A_c\widetilde X_T
 -\widehat B_cU_T\|_{\Sigma^{-1},\mathrm F}^2
 =
 \|Z_T-A_c X_T-B_c U_T\|_{\Sigma^{-1},\mathrm F}^2.
\]
Taking the infimum over the bijective correspondence between the candidate models gives
$S_T(\widetilde D_T)=S_T(D_T)$.

(Necessity)
Suppose that \eqref{eq:perfect_stealth_definition} holds for any finite-length data.

\indent\textit{(i) Non-singularity of $T_{x,t}$:}
First, let $t\geq1$.
Fix $T\geq\max\{2,t\}$ and suppose that there exists $v\neq0$ such that
$T_{x,t}v=0$.
Consider data for which $x_t=v$ is the only nonzero state and all other states and inputs are zero.
The residual at time $t-1$ is $v$ for every candidate model, and hence
$
 S_T(D_T)\geq\|v\|_{\Sigma^{-1},\mathrm F}^2>0.
$
On the other hand, let $D_T^0$ denote the data for which all states and inputs are zero.
Since $T_{x,t}v=0$, the two datasets become identical after transformation.
Moreover, $S_T(D_T^0)=0$, and therefore
\eqref{eq:perfect_stealth_definition} implies
$S_T(D_T)=0$, which is a contradiction.
Thus, $T_{x,t}$ is invertible for all $t\geq1$.

Next, suppose that there exists $v\neq0$ satisfying $T_{x,0}v=0$.
By \eqref{eq:no_common_action_direction}, there exists
$A\in\mathcal A$ such that $Av\neq0$.
For $T=2$, define the data $D^1$ by
$
 x_0=v,
 x_{s+1}=Ax_s,
 u_s=0,
 s=0,1
$.
Then $S_2(D^1)=0$.
Let $D^2$ be obtained from $D^1$ by replacing only $x_0$ with zero.
Since the residual at time $0$ is $Av$ for every candidate model,
\begin{equation}
 S_2(D^2)\geq\|Av\|_{\Sigma^{-1},\mathrm F}^2>0.
 \label{eq:ineq_S2D2}
\end{equation}
Since $T_{x,0}v=0$, we have
$\widetilde D^1=\widetilde D^2$.
Therefore, \eqref{eq:perfect_stealth_definition} gives
$S_2(\widetilde D^2)=S_2(\widetilde D^1)=0$, which contradicts
\eqref{eq:ineq_S2D2}.
Hence,
\begin{equation}
 T_{x,t}\in\GL(n),
 \qquad
 t=0,1,\ldots .
 \label{eq:T_invertible}
\end{equation}

\indent\textit{(ii) Vanishing of the bias terms:}
For any $A\in\mathcal A$, define
$
 x_t=-A^t T_{x,0}^{-1}b_{x,0},
 u_t=0
$.
Then $x_{t+1}=Ax_t$, and hence
$S_T(D_T)=0$.
Since $\widetilde x_0=0$ after transformation, the residual at time $0$ is
$\widetilde x_1$ independently of the candidate model.
Thus, \eqref{eq:perfect_stealth_definition} implies
$\widetilde x_1=0$, and therefore for any $A\in\mathcal A$, $T_{x,1}A T_{x,0}^{-1}b_{x,0}=b_{x,1}$.
Taking the difference between any $A_1,A_2\in\mathcal A$ and using
\eqref{eq:T_invertible} gives
\[
 (A_1-A_2)T_{x,0}^{-1}b_{x,0}=0,
 \qquad
 \forall A_1,A_2\in\mathcal A.
\]
It follows from \eqref{eq:no_common_action_direction} that
$b_{x,0}=0$.

Furthermore, consider the all-zero dataset.
The transformed state sequence is
$b_{x,0},b_{x,1},\ldots,b_{x,T}$, and its anomaly score is also zero.
Since $b_{x,0}=0$, the residual at time $0$ is $b_{x,1}$, implying
$b_{x,1}=0$.
Similarly, if $b_{x,0}=\cdots=b_{x,k}=0$, the residual at time $k$ is
$b_{x,k+1}$.
Thus, by induction, we obtain
\begin{equation}
 b_{x,t}=0,
 \qquad
 t=0,1,\ldots.
 \label{eq:b_zero}
\end{equation}

\indent\textit{(iii) Time invariance of $T_{x,t}$:}
Fix arbitrary $T\geq2$ and $c\in\R^n$, and define
\begin{equation}
 h_t:=(T_{x,t}-T_{x,0})c,
 \qquad
 t=0,\ldots,T.
 \label{eq:h_definition}
\end{equation}
We first show that, for any state sequence $(y_0,\ldots,y_T)$,
\begin{equation}
 S_{e_1,T}(y_0+h_0,\ldots,y_T+h_T)
 =
 S_{e_1,T}(y_0,\ldots,y_T).
 \label{eq:transformed_translation}
\end{equation}

By \eqref{eq:T_invertible}, let
$z_t:=T_{x,t}^{-1}y_t$.
From Lemma~\ref{lem:common_translation},
\[
 S_{e_1,T}(z_0+c,\ldots,z_T+c)
 =
 S_{e_1,T}(z_0,\ldots,z_T).
\]
Applying \eqref{eq:b_zero} and anomaly-score invariance to both datasets gives
$
 S_{e_1,T}(y_0+T_{x,0}c,\ldots,y_T+T_{x,T}c)
 =
 S_{e_1,T}(y_0,\ldots,y_T).
$
Adding the common translation $-T_{x,0}c$ to all states on the left-hand side and applying Lemma~\ref{lem:common_translation} again yields
\eqref{eq:transformed_translation}.

Next, fix any $A_0\in\mathcal A$ and define
$
 z_0=0,
 z_{t+1}=A_0z_t+h_1.
$
Choosing the candidate model
$A_c=A_0$, $B_c=h_1e_1^\top$ gives
$S_{e_1,T}(z_0,\ldots,z_T)=0$.
Let $w_t:=z_t-h_t$.
Then, by \eqref{eq:transformed_translation},
$S_{e_1,T}(w_0,\ldots,w_T)=0$.
Moreover, since $h_0=0$, $z_0=0$, and $z_1=h_1$, we have
$w_0=w_1=0$.
Applying Lemma~\ref{lem:zero_propagation} with $i=0$ and $j=1$ yields
$w_2=0$.
On the other hand, since
$z_2=A_0h_1+h_1$,
\begin{equation}
 A_0h_1=h_2-h_1.
 \label{eq:h1_relation}
\end{equation}

Since $A_0\in\mathcal A$ is arbitrary,
$(A_1-A_2)h_1=0$
for any $A_1,A_2\in\mathcal A$.
By \eqref{eq:no_common_action_direction},
$h_1=0$, and \eqref{eq:h1_relation} then gives
$h_2=0$.

Suppose further that
$h_0=\cdots=h_k=0$ for some $k\in\{2,\ldots,T-1\}$.
Applying \eqref{eq:transformed_translation} to the all-zero state sequence yields
$S_{e_1,T}(h_0,\ldots,h_T)=0$.
Since $h_0=h_1=h_k=0$,
Lemma~\ref{lem:zero_propagation} with $i=0$ and $j=k$ gives
$h_{k+1}=0$.
Therefore, by induction,
$h_t=0$ for $t=0,\ldots,T$.
Since $c$ is arbitrary,
$T_{x,t}=T_{x,0}$ for $t=0,\ldots,T$.
Since $T\geq2$ is arbitrary, there exists $T_x\in\GL(n)$ such that
\begin{equation}
 T_{x,t}=T_x,
 \qquad
 t=0,1,\ldots .
 \label{eq:T_time_invariant}
\end{equation}

\indent\textit{(iv) Covariance invariance:}
Let $T=2$ and set
$x_0=x_1=0,x_2=v,u_0=u_1=0$.
Then
$S_2(D_2)=\|v\|_{\Sigma^{-1},\mathrm F}^2$.
By \eqref{eq:b_zero} and
\eqref{eq:T_time_invariant},
$S_2(\widetilde D_2)=\|T_xv\|_{\Sigma^{-1},\mathrm F}^2$.
Thus, for any $v\in\R^n$, \eqref{eq:perfect_stealth_definition} gives
$
 \|T_xv\|_{\Sigma^{-1},\mathrm F}^2 = \|v\|_{\Sigma^{-1},\mathrm F}^2.
$
Therefore, $ T_x^\top {\Sigma^{-1}}T_x={\Sigma^{-1}}$.

\indent\textit{(v) Invariance of the prior-knowledge set:}
First, fix
$A_0\in\mathcal A\cap\mathcal C_{m+1}$.
By \eqref{eq:Cm1}, there exist $g_0,\ldots,g_m\in\R^n$ such that, with
$
 G_0:=
 \begin{bmatrix}
  g_0 & \cdots & g_m
 \end{bmatrix}
$,
we have
\begin{equation}
 \operatorname{rank}
 \begin{bmatrix}
  G_0 & A_0G_0 & \cdots & A_0^{n-1}G_0
 \end{bmatrix}
 =n.
 \label{eq:rank_identification}
\end{equation}

For $i=1,\ldots,m$, let
$e_i\in\R^m$ denote the $i$th standard basis vector.
Define $B_0\in\R^{n\times m}$ by
\begin{equation}
 B_0e_i
 :=
 g_i-A_0^{n+1}g_{i-1},
 \qquad
 i=1,\ldots,m.
 \label{eq:B0_definition}
\end{equation}
Also, let
$
 T_\star:=(m+1)(n+1)-1,
 \tau_i:=i(n+1),
 i=0,\ldots,m.
$
For each $i=0,\ldots,m$ and $k=0,\ldots,n$, define
$x_{\tau_i+k}:=A_0^kg_i$.
Within each interval, let
$
 u_{\tau_i+k}=0,
 k=0,\ldots,n-1,
$
and between consecutive intervals, let
$
 u_{\tau_i+n}=e_{i+1},
 i=0,\ldots,m-1
$.
Within each interval,
$x_{\tau_i+k+1}=A_0x_{\tau_i+k}+B_0u_{\tau_i+k}$.
At the transitions between intervals, \eqref{eq:B0_definition} gives
\[
 A_0x_{\tau_i+n}+B_0u_{\tau_i+n}
 =
 A_0^{n+1}g_i+B_0e_{i+1}
 =
 g_{i+1} =
 x_{\tau_i+n+1}.
\]
Hence, this single dataset satisfies
$x_{t+1}=A_0x_t+B_0u_t$
at every $t$, and thus
$S_{T_\star}(D_{T_\star})=0$.
By anomaly-score invariance,
\begin{equation}
 S_{T_\star}(\widetilde D_{T_\star})=0.
 \label{eq:transformed_zero_score}
\end{equation}
Within each interval, the input is zero.
Furthermore, by \eqref{eq:rank_identification} and
$T_x\in\GL(n)$,
\begin{equation}
 \operatorname{span}
 \left\{
 T_xA_0^kg_i:
 i=0,\ldots,m,\;
 k=0,\ldots,n-1
 \right\}
 =
 \R^n.
 \label{eq:transformed_state_span}
\end{equation}
Thus, from \eqref{eq:transformed_zero_score},
\eqref{eq:transformed_state_span}, and
Assumption~\ref{ass:prior_knowledge},
there exists $\widehat A\in\mathcal A$ such that, on each zero-input interval,
\begin{equation}
 T_xA_0^{k+1}g_i
 =
 \widehat A T_xA_0^kg_i,
 \quad
 i=0,\ldots,m,\;
 k=0,\ldots,n-1.
 \label{eq:zero_residual_identification}
\end{equation}
On the other hand, define
$
 M:=T_xA_0T_x^{-1}.
$
Then
$MT_xA_0^kg_i = T_xA_0^{k+1}g_i$.
Hence, \eqref{eq:zero_residual_identification} implies
$(\widehat A-M)T_xA_0^kg_i=0$
for all $i$ and $k$.
By \eqref{eq:transformed_state_span},
$\widehat A=M=T_xA_0T_x^{-1}$.
Since $\widehat A\in\mathcal A$,
\begin{equation}
 A_0\in\mathcal A\cap\mathcal C_{m+1}
 \quad\Longrightarrow\quad
 T_xA_0T_x^{-1}\in\mathcal A.
 \label{eq:prior_on_dense_set}
\end{equation}

Next, by Assumption~\ref{ass:stealth_characterization},
$\mathcal A\cap\mathcal C_{m+1}$ is dense in $\mathcal A$.
The similarity transformation
$A\mapsto T_xAT_x^{-1}$
is continuous, and $\mathcal A$ is closed.
Therefore, \eqref{eq:prior_on_dense_set} holds for every
$A\in\mathcal A$, i.e.,
\begin{equation}
 A\in\mathcal A
 \quad\Longrightarrow\quad
 T_xAT_x^{-1}\in\mathcal A.
 \label{eq:prior_forward}
\end{equation}

Finally, since $T_x$ is invertible and anomaly-score invariance holds for every finite-length dataset, the inverse transformation
$x_t\mapsto T_x^{-1}x_t$
also preserves the anomaly score.
Applying the same argument to $T_x^{-1}$ gives
$
 A\in\mathcal A
 \Longrightarrow
 T_x^{-1}AT_x\in\mathcal A.
$
Equivalently,
$
 T_xAT_x^{-1}\in\mathcal A
 \Longrightarrow
 A\in\mathcal A.
$
Combining this implication with \eqref{eq:prior_forward} yields
$
 A\in\mathcal A
 \iff
 T_xAT_x^{-1}\in\mathcal A.
$
\end{proof}

Theorem 1 shows that requiring anomaly-score preservation for arbitrary data restricts the class of admissible affine attacks. Specifically, 
condition~\eqref{eq:time_invariant_sensor_attack} shows that any affine transformation satisfying \textit{S}-invariant stealthiness must be time invariant and have no bias term.
Condition~\eqref{eq:weighted_isometry_condition} states that the sum of squared Mahalanobis distances of the state residuals is preserved under the transformation $T_x$.
Condition~\eqref{eq:prior_invariance_plus} ensures that the candidate models explaining the transformed data remain within the prior-knowledge set.
Many prior-knowledge sets in control satisfy \eqref{eq:prior_invariance_plus} for all $T_x\in\GL(n)$. For example, similarity transformations preserve the spectral radius
$\rho(T_xA_cT_x^{-1})=\rho(A_c)$.

\begin{remark}[On the Assumed Prior-Knowledge Set]
The first two conditions in Assumption~\ref{ass:stealth_characterization} may not be restrictive in practice.
First, if the first condition does not hold, there exist
$x\in\mathbb R^n\setminus\{0\}$ and $y\in\mathbb R^n$
such that every $A\in\mathcal A$ satisfies the linear equation $Ax=y$.
Thus, this condition means that the prior-knowledge set is not contained in the solution set of any equation $Ax=y$ with $x\neq0$.
Next, the second condition requires that matrices whose eigenvalues have geometric multiplicity at most $m+1$ be dense in $\mathcal{A}$~\cite{zhou2017minimal}. As shown in Theorem \ref{theorem:clousure_condition} in the Appendix, this condition holds for a broad class of prior-knowledge sets, including common convex uncertainty sets under mild conditions.
\end{remark}

\begin{remark}[Apparent System Model]
The proposed transformation can be interpreted as spoofing the system dynamics.
To clarify this interpretation, consider the noise-free relation $Z=AX+BU$.
Under the transformation
$\widetilde X=T_xX$, $\widetilde Z=T_xZ$, and $\widetilde U=U$,
we obtain
\begin{equation}
 \widetilde Z
 =
 (T_xAT_x^{-1})\widetilde X+(T_xB)\widetilde U.
\end{equation}
Thus, the manipulated data are consistent with the apparent system
$
 \widetilde A=T_xAT_x^{-1}, \widetilde B=T_xB
$.
Hence, the proposed transformation can be viewed as spoofing the system model perceived by the designer.
For example, choosing $T_x=-I$ yields
$\widetilde A=A$ and $\widetilde B=-B$,
and particularly in the LQR design, replacing $B$ with $-B$ reverses the sign of the resulting feedback gain, i.e., $\widetilde K=-K$.
\end{remark}

\section{Attack Design over Orthogonal Group}\label{sec:optimization}

\subsection{Formulation as Optimization over Riemannian Manifold}
In this section, we develop a method for designing data poisoning attacks using the characterization from the previous section.
First, we show that the transformations in Theorem~\ref{thm:perfect_stealth} can be parameterized using orthogonal matrices.
\begin{prop}[Parameterization by Orthogonal Matrices]\label{thm:weighted_orthogonal_parameterization}
A transformation matrix $T_x\in\GL(n)$ satisfies $T_x^\top\Sigma^{-1}T_x=\Sigma^{-1}$
if and only if there exists $\Theta\in\Ogroup(n)$ such that it can be expressed as
\begin{equation}
 T_x(\Theta)=\Psi^{-1}\Theta\Psi
 \label{eq:weighted_orthogonal_parameterization}
\end{equation}
where $\Psi\in\GL(n)$ is the Cholesky factor of $\Sigma^{-1}$ satisfying $\Sigma^{-1}=\Psi^\top\Psi$.
\end{prop}

\begin{proof}
First, suppose that $T_x$ satisfies $T_x^\top\Sigma^{-1}T_x=\Sigma^{-1}$.
Let
$\Theta:=\Psi T_x\Psi^{-1}$.
Then
$\Theta^\top\Theta=I$.
Therefore, $\Theta\in\Ogroup(n)$, and
$T_x$ can be expressed as $T_x=\Psi^{-1}\Theta\Psi$.
Conversely, suppose that
$T_x=\Psi^{-1}\Theta\Psi$ for some $\Theta\in\Ogroup(n)$.
Then,
\begin{align*}
 T_x^\top\Sigma^{-1}T_x
 &=
 \left(\Psi^{-1}\Theta\Psi\right)^\top
 \Psi^\top\Psi
 \left(\Psi^{-1}\Theta\Psi\right) \nonumber\\
 &=
 \Psi^\top\Theta^\top\Theta\Psi
 =
 \Psi^\top\Psi
 =
 \Sigma^{-1}.
\end{align*}
\end{proof}

By Theorem~\ref{thm:weighted_orthogonal_parameterization}, \textit{S}-invariant stealthy sensor attack design can be treated as a problem of selecting a matrix $\Theta\in\Ogroup(n)$ on the orthogonal group.
Under this parameterization, we propose a method for maximizing the attacker's objective function $J_{\mathrm{atk}}(\design(\mathcal{T}(D));A,B)$ while preserving \textit{S}-invariant stealthiness.
Here, the condition \eqref{eq:prior_invariance_plus} on the prior-knowledge set depends on $\Theta$, but in this section, we impose the following assumption on the prior knowledge.
\begin{assumption}\label{ass:prior_knowledge_sufficient}
For every $\Theta\in\Ogroup(n)$, the corresponding transformation $T_x(\Theta)$ satisfies \eqref{eq:prior_invariance_plus}.
\end{assumption}

As discussed in Section~\ref{sec:stealth}, many prior-knowledge sets naturally used in control satisfy Assumption~\ref{ass:prior_knowledge_sufficient}.

Using \eqref{eq:weighted_orthogonal_parameterization}, we write the manipulated data as
$
 D(\Theta):= \left(T_x(\Theta)X,\ T_x(\Theta)Z,\ U\right).
$
Then, the data poisoning attack design problem can be formulated as follows.
\begin{prop}[Reformulation using Orthogonal Group]
\label{thm:apparent_system}
Suppose that Assumptions 1--3 hold for the prior knowledge.
Define the attack objective function as
$
 \Phi(\Theta;A,B)
 :=
 J_{\mathrm{atk}}\left(\design(D(\Theta));A,B\right)
$.
Then, the \textit{S}-invariant stealthy attack design problem~\eqref{eq:conceptual_attack_problem} is equivalent to
\begin{equation}
 \underset{\Theta\in\Ogroup(n)}{\operatorname{maximize}}
 \quad
 \Phi(\Theta;A,B).
 \label{eq:orthogonal_attack_optimization}
\end{equation}
\end{prop}
\begin{proof}
    From the above discussion, the equivalence follows by considering the correspondence $T_{x}=\Psi^{-1}\Theta\Psi$.
\end{proof}

We now clarify how the proposed attack can be applied to previously unseen data.  While S-invariant stealthiness is guaranteed independently of the collected data, the attack objective in~\eqref{eq:orthogonal_attack_optimization} depends on the dataset used for optimization. Thus, the resulting transformation preserves the anomaly score even when applied to previously unseen data, although its attack effectiveness on such data is not necessarily guaranteed.

We next develop a method for solving the optimization problem ~\eqref{eq:orthogonal_attack_optimization}. Since the feasible set is the orthogonal group $\Ogroup(n)$, it is difficult to directly apply standard constrained optimization algorithms. 
Therefore, we regard the orthogonal group as a Riemannian manifold and adopt a framework for optimization on manifolds~\cite{absil2008optimization,edelman1998geometry}.

At each iteration of a gradient method on a manifold, the Euclidean gradient is computed and projected onto the tangent space at the current point to obtain the Riemannian gradient.
Finally, a mapping called a retraction is used to bring the updated point back onto the orthogonal group.

Let the Euclidean gradient at the current point $\Theta_k \in \Ogroup(n)$ be
$
 G_k=\nabla_{\Theta}\Phi(\Theta_k;A,B).
$
The specific method for computing this gradient depends on the structure of the mapping $\design$, but we assume that the gradient can be evaluated.
For example, when $\design$ is given as the solution to an optimization problem, the gradient may be derived using implicit differentiation based on the optimality conditions and the chain rule~\cite{sasahara2026adversarial}.
Furthermore, this gradient can be computed efficiently by applying automatic differentiation~\cite{NEURIPS2019_9ce3c52f}.
The Riemannian gradient $H_k$, obtained by projecting the Euclidean gradient $G_k$ onto the tangent space of the orthogonal group, is given by
$
 H_k
 =
 G_k-\Theta_k\sym(\Theta_k^\top G_k).
$
For a step size $\alpha>0$, let
$Y_k=\Theta_k+\alpha H_k$
and write its singular value decomposition as
$Y_k=P_k\Lambda_kQ_k^\top$.
Then, the update
$
\Theta_{k+1}=P_kQ_k^\top
$
projects $Y_k$ onto its nearest orthogonal matrix.

\subsection{Extension to Attacks Close to the Original Data}
\label{subsec:regularized_attack}
In the optimization problem \eqref{eq:orthogonal_attack_optimization}, the attack objective function $\Phi(\Theta;A,B)$ is maximized over the class of transformations satisfying \textit{S}-invariant stealthiness.
However, to more reliably prevent the attack from being detected, one can consider suppressing excessive deviations from the original data in addition to preserving anomaly-score invariance.
Therefore, in this subsection, we introduce a regularization term into the objective function to reduce the magnitude of data manipulation.

The difference between the clean data and the manipulated data is
$
 D(\Theta)-D
 =
 \left(
 (T_x(\Theta)-I)X,\ 
 (T_x(\Theta)-I)Z,\ 
 0
 \right).
$
Therefore, the magnitude of data manipulation can be bounded as
\begin{equation*}
 \norm{D(\Theta)-D}_{\mathrm F}^2 \leq
 \norm{T_x(\Theta)-I}_{\mathrm F}^2
 \left(
 \norm{X}_{\mathrm F}^2+\norm{Z}_{\mathrm F}^2
 \right).
 \label{eq:data_difference_upper_bound}
\end{equation*}
Thus, by bringing $T_x(\Theta)$ closer to the identity matrix,
the manipulated data can also be expected to become closer to the original data.
Based on this observation, we consider an optimization problem with the following regularized objective function involving a parameter $\gamma\geq0$:
\begin{equation*}
 \underset{\Theta\in\Ogroup(n)}{\operatorname{maximize}}
 \quad
  \Phi_{\gamma}(\Theta;A,B)
 =
 \Phi(\Theta;A,B)
 -
 \gamma
 \norm{T_x(\Theta)-I}_{\mathrm F}^2.
\label{eq:orthogonal_attack_optimization_reg}
\end{equation*}
Here, \textit{S}-invariant stealthiness is preserved by $\Theta\in\Ogroup(n)$, regardless of whether regularization is applied.
Therefore, replacing $\Phi$ with $\Phi_{\gamma}$ allows the same gradient method on the orthogonal group to be applied.

\section{Numerical Examples}\label{sec:numerical}
In this section, we numerically verify that the proposed attack can degrade the closed-loop performance of the controller while preserving the anomaly score.
As a baseline for comparison, we add independent uniform noise sampled from \([-\varepsilon\|X\|_{\max},\varepsilon\|X\|_{\max}]\) to each entry of the state data.
We also compare the proposed method with I-DGSM~\cite{sasahara2026adversarial}, which searches for data manipulations using a projected gradient method subject to the constraint $\|\Delta X\|_{\max}<\epsilon \|X\|_{\max}$.

We used the Suspension System (SS) model as the target system~\cite{tilbury2025ctms}.
The sampling period was set to $0.1$.
Its dimensions are $(n,m)=(4,2)$, its open-loop spectral radius is $\rho(A)=0.99$, and the closed-loop spectral radius with the controller designed from the original data is $\rho(A+BK)= 0.924$.
The designer employs the direct data-driven control method in~\cite{sasahara2026adversarial}.
The regularization parameter is set to $\lambda=10^5$, and the LQR controller weights are set to $Q=I, R=0.1I$.
In addition, the data targeted by the attack are assumed to be noise-free, and the data length is set to $T=15$.
The weighting matrix used to compute the anomaly score is set to $\Sigma=10^{-2}I$, and the designer's prior-knowledge sets are $\mathcal A = \R^{n\times n}$ and $\mathcal B = \R^{n\times m}$.

The attack objective is to maximize the closed-loop spectral radius $J_{\mathrm{atk}}(\cdot)=\rho(A+B\widetilde K)$ when the controller $\widetilde K$ designed from the manipulated data is implemented on the true system.
Since $\rho(\cdot)$ is not differentiable when the dominant eigenvalue is not simple, we use smoothing via the Gelfand approximation $\|M^{q}\|_{\mathrm{F}}^{1/q}$ with $q=20$~\cite{sasahara2026adversarial}.
The step sizes in the optimization algorithms are set to $\epsilon \|D\|_{\max}$ for I-DGSM and $0.2$ for the proposed method.
The number of iterations is set to 20, and 80 trials are performed for each parameter value.
The evaluation metrics are the closed-loop spectral radius $\rho(A+B\widetilde K)$, the anomaly score $\score(\widetilde D)$, and the magnitude of manipulation relative to the original data, $\norm{\widetilde D-D}_{\mathrm F}$.
For I-DGSM, the coefficient $\epsilon$ controlling the magnitude of data manipulation is used as a parameter, while the perturbation amplitude $\epsilon$ is used for the random-noise method.
For the proposed method, the regularization parameter $\gamma$ is varied.

\subsubsection{Representative Results for SS Model}
First, we verify that the proposed method can degrade closed-loop performance while preserving the anomaly score for the SS model.
Table~\ref{tab:ss_representative} displays representative attack results.
\begin{table}[tb]
\centering
\caption{Attack results for the SS model}
\label{tab:ss_representative}
\resizebox{\linewidth}{!}{%
\begin{tabular}{ccccc}
\toprule
Method
& Parameter
& $\rho(A+B\widetilde K)$
& $\score(\widetilde D)$
& $\norm{\widetilde D-D}_{\mathrm F}$\\
\midrule
Proposed
& $\gamma=0$
& 2.674
& $2.34\times10^{-23}$
& 23.17\\
Random
& $\epsilon=1\times10^{-1}$
& 0.948
& $6.02\times10^{4}$
& 16.36\\
I-DGSM
& $\epsilon=1\times10^{-2}$
& 1.143
& $3.01\times10^{3}$
& 1.46\\
I-DGSM
& $\epsilon=2\times10^{-1}$
& 1.234
& $2.21\times10^{4}$
& 21.35\\
\bottomrule
\end{tabular}%
}
\end{table}
Observe that the proposed method without regularization increases the closed-loop spectral radius to $2.674$ while keeping the anomaly score nearly zero at $2.34\times10^{-23}$.
This confirms that the proposed method significantly degrades closed-loop performance while preserving \textit{S}-invariant stealthiness.
For the random-noise method, although the magnitude of data manipulation is comparable to that of the proposed method, the closed-loop spectral radius remains below $1$, and the system is not destabilized.
Moreover, the anomaly score increases to $6.02\times10^{4}$.
For I-DGSM with $\epsilon=1\times10^{-2}$, the magnitude of data manipulation is small, but the anomaly score is $3.01\times10^{3}$.
In addition, the spectral radius is $1.143$, indicating that the system is destabilized, although this value is smaller than that achieved by the proposed method.
For $\epsilon=2\times10^{-1}$, the magnitude of data manipulation is comparable to that of the proposed method, but the anomaly score increases to $2.21\times10^{4}$.
Furthermore, although the spectral radius increases to $1.234$, it remains smaller than that achieved by the proposed method.

These results show that the proposed method can cause instability, while the manipulation cannot be detected based on the anomaly score.
In addition, the proposed method achieves higher attack effectiveness than the other methods with comparable magnitudes of data manipulation.

\subsubsection{Effect of Regularization Parameter}
As shown so far, the proposed method
degrades closed-loop performance while preserving the anomaly score.
However, large deviations from the original data can be observed in Fig.~\ref{fig:trajectory_comparison}.
We therefore vary the regularization parameter $\gamma$
to examine the trade-off between attack effectiveness
and manipulation magnitude.
Table~\ref{tab:ss_param_sweep_updated} confirms
a clear trade-off between these two objectives.

\begin{table}[tb]
\centering
\caption{Regularization trade-off in the proposed method}
\label{tab:ss_param_sweep_updated}
\resizebox{\linewidth}{!}{%
\begin{tabular}{ccccccc}
\toprule
$\gamma$
& 0 & 0.5 & 1 & 2 & 3 & 4 \\
\midrule
$\rho(A+B\widetilde K)$
& 2.674 & 2.123 & 1.657 & 0.990 & 0.955 & 0.950 \\
$\norm{\widetilde D-D}_{\mathrm F}$
& 23.17 & 20.60 & 15.55 & 8.99 & 5.84 & 4.23 \\
\bottomrule
\end{tabular}%
}
\end{table}

\begin{figure}[tb]
\centering
\includegraphics[width=.9\linewidth]{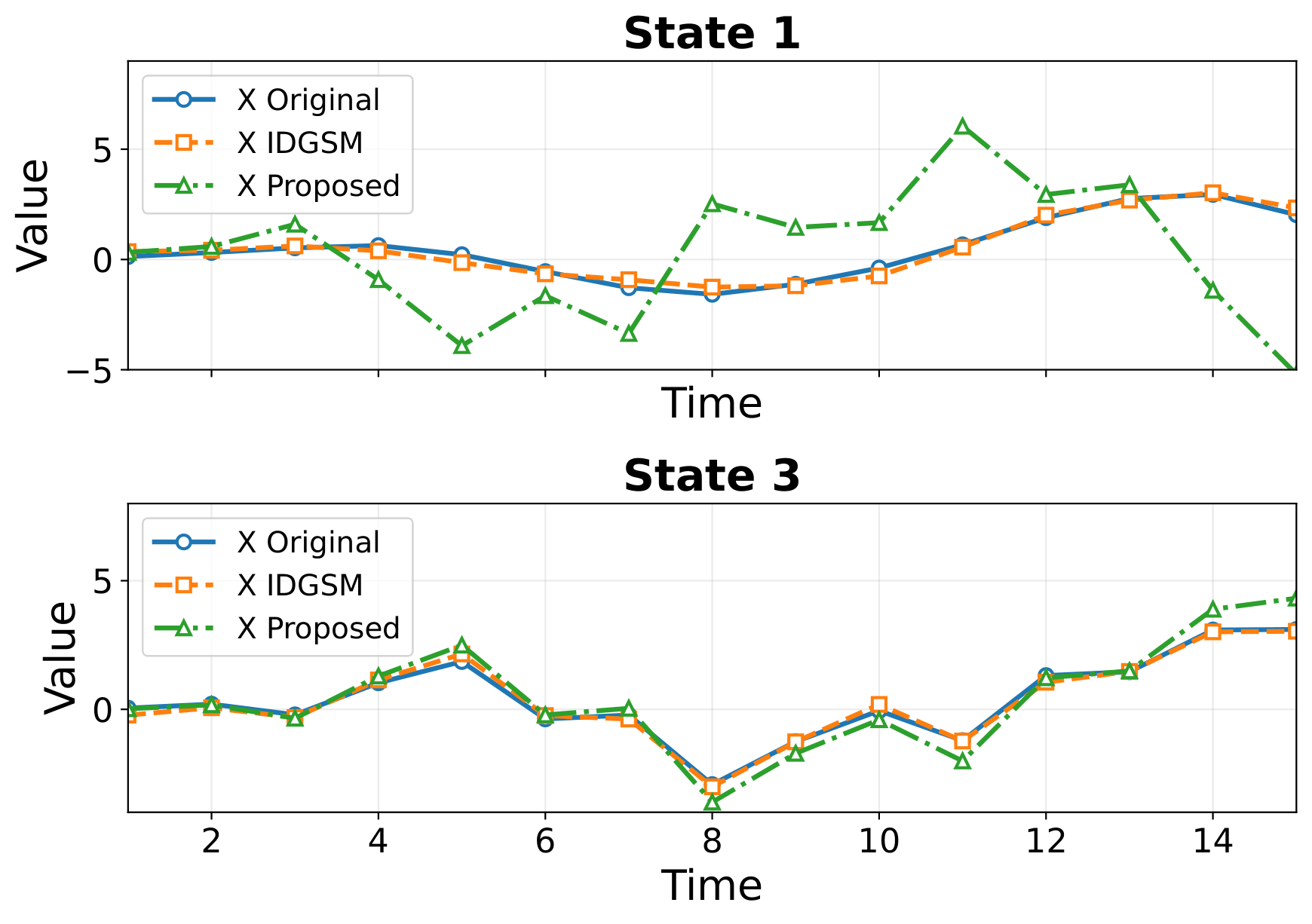}
\caption{Examples of data after the attack for the SS model. The I-DGSM results are plotted for $\epsilon=1\times10^{-2}$.}
\label{fig:trajectory_comparison}
\end{figure}

\begin{figure}[tb]
\centering
\includegraphics[width=.9\linewidth]{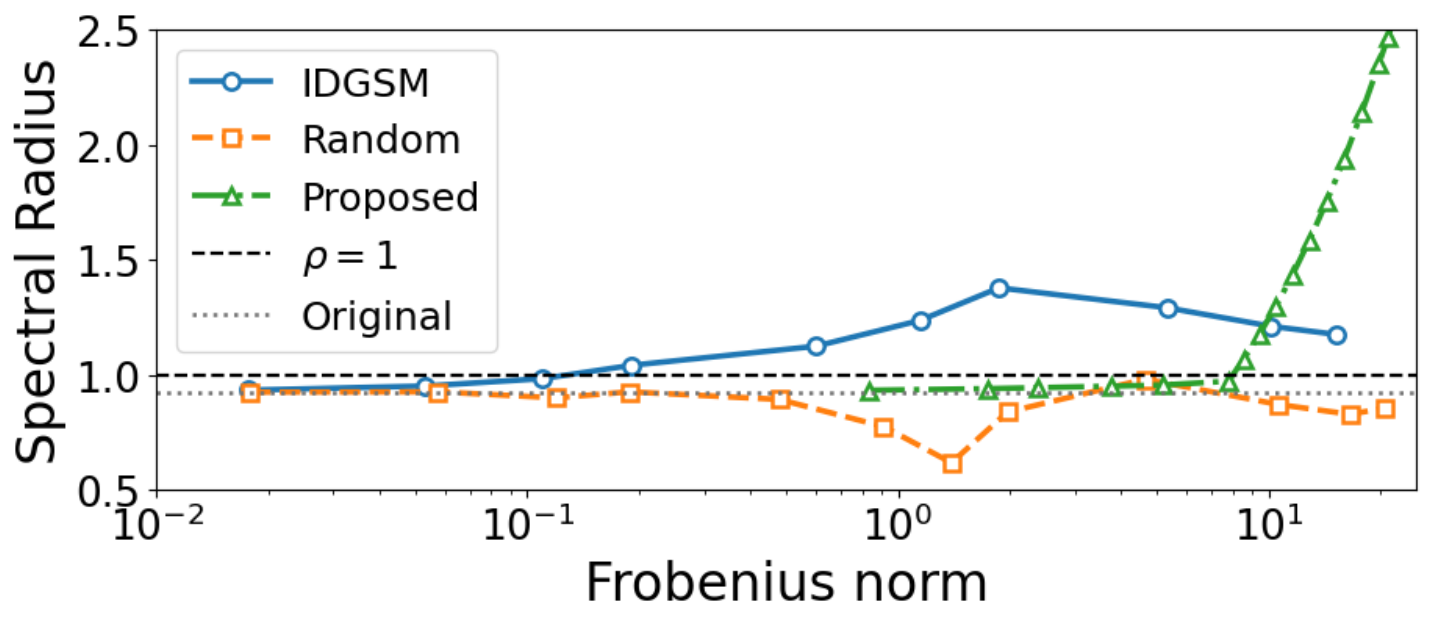}
\caption{Relationship between the magnitude of data manipulation and the spectral radius.}
\label{fig:fro_spectral}
\end{figure}

Next, we examine the relationship between attack effectiveness and the magnitude of data manipulation when the parameters are varied as plotted in Fig.~\ref{fig:fro_spectral}.
For small magnitudes of data manipulation, I-DGSM produces a larger closed-loop spectral radius than the proposed method and thus achieves higher attack effectiveness.
However, for larger magnitudes, the proposed method achieves a higher attack effectiveness.
We also observe that the behavior of the proposed method differs significantly between the range in which the closed-loop spectral radius is below $1$ and the range in which it exceeds $1$ and destabilizes the system.


\section{Conclusion}\label{sec:conclusion}
This paper has investigated stealthy sensor-data poisoning attacks on data-driven controller design.
We have characterized the transformations that preserve the anomaly score and established time-invariant weighted orthogonal transformations as their necessary and sufficient form.
Based on this characterization, we have developed an attack optimization method over the orthogonal group.
Numerical experiments have demonstrated that the proposed method can degrade closed-loop performance while preserving the anomaly score.
Future work includes attacks on subsets of sensors, more restrictive prior knowledge, and improved optimization on the orthogonal group.

\appendix

In this appendix, we provide sufficient conditions for the density requirement in the second condition of Assumption \ref{ass:stealth_characterization} to hold. We first characterize the set $\mathcal C_p$ defined in \eqref{eq:Cm1}.

\begin{lemma}
\label{lemma:AppendixA}
There exists a nonzero polynomial $\Psi_p(A)$ such that
$A\in\mathcal C_p \Leftrightarrow \Psi_p(A)>0.$
\end{lemma}

\begin{proof}
Let $K_p(A,G):=[\,G\ AG\ \cdots\ A^{n-1}G\,]$,
and denote its $n\times n$ minors by $q_\ell(A,G)$.
Expanding them in the entries of $G$ gives
$q_\ell(A,G)=\sum_\alpha c_{\ell,\alpha}(A)G^\alpha$,
where each $c_{\ell,\alpha}(A)$ is a polynomial in $A$.
Since a polynomial vanishes identically if and only if
all its coefficients vanish,
\begin{equation*}
\begin{aligned}
A\notin\mathcal C_p
\Longleftrightarrow
\rank K_p(A,G)<n, \forall G
&\Longleftrightarrow
q_\ell(A,G)=0, \forall G,\ell\\
&\Longleftrightarrow
c_{\ell,\alpha}(A)=0, \forall\ell,\alpha.
\end{aligned}
\end{equation*}
Thus, defining
$\Psi_p(A):=\sum_{\ell,\alpha}c_{\ell,\alpha}(A)^2$
yields the desired equivalence.
Moreover, $\Psi_p$ is not identically zero,
since any matrix with distinct eigenvalues belongs to $\mathcal C_p$.
\end{proof}

Suppose that there exists $\mathcal M\subset\R^{n\times n}$ such that the following three conditions hold: (M1): $\mathcal M\subseteq\mathcal A\subseteq\overline{\mathcal M}.$ (M2): $\mathcal M$ is a connected real-analytic submanifold. (M3): $\mathcal M$ satisfies $\mathcal M\cap\mathcal C_p\neq\varnothing$.



\begin{theorem}
\label{theorem:clousure_condition}
Under (M1)--(M3), $\mathcal A\cap\mathcal C_p$ is dense in $\mathcal A$.
\end{theorem}

\begin{proof}
By (M3) and the preceding lemma, there exists
$A_\star\in\mathcal M$ such that $\Psi_p(A_\star)>0$.
Thus, $\Psi_p|_{\mathcal M}$ is a nonzero
real-analytic function.
Suppose that $\mathcal M\cap\mathcal C_p$
is not dense in $\mathcal M$.
Then, there exists a nonempty open subset
$V\subseteq\mathcal M$ such that
$V\cap\mathcal C_p=\varnothing$.
By Lemma~\ref{lemma:AppendixA},
$\Psi_p(A)=0$ for all $A\in V$.
Since $\mathcal M$ is connected,
the identity theorem for real-analytic functions
implies $\Psi_p|_{\mathcal M}\equiv0$,
contradicting $\Psi_p(A_\star)>0$.
Hence, $\mathcal M\cap\mathcal C_p$ is dense
in $\mathcal M$.
By (M1),
$
\mathcal{A}
\subseteq\overline{\mathcal M}
=\overline{\mathcal M\cap\mathcal C_p}
\subseteq\overline{\mathcal A\cap\mathcal C_p},
$
which completes the proof.
\end{proof}

These conditions hold for common convex sets,
including sets with entrywise interval uncertainty
$\mathcal A=\{A:\underline A_{i,j}\le A_{i,j}
\le\overline A_{i,j},\ \forall i,j\}$,
ellipsoidal uncertainty
$\mathcal A=\{A:(A-A_0)^\top Q^{-1}(A-A_0)\preceq I\}$
with $Q\succ0$,
and symmetric matrices
$\mathcal A=\{A:A=A^\top\}$,
provided that $\mathcal A\cap\mathcal C_p\neq\varnothing$.

\bibliographystyle{IEEEtran}
\bibliography{references}

\end{document}